\documentclass[abstract=true, DIV=16, parskip=half]{scrartcl}

\usepackage[noEnd,commentColor=black]{algpseudocodex}
\usepackage{amsmath}
\usepackage{amssymb}
\usepackage{amsthm}
\usepackage[mathscr]{euscript}
\usepackage[shortlabels]{enumitem}
\usepackage{graphicx} 
\usepackage{subcaption}
\usepackage{thmtools}
\usepackage{todonotes}
\usepackage{xcolor}
\usepackage{microtype}

\makeatletter
\def\@fnsymbol#1{\ensuremath{\ifcase#1\or \!\;\or \!\;\or \ddagger\or
   \mathsection\or \mathparagraph\or \|\or **\or \dagger\dagger
   \or \ddagger\ddagger \else\@ctrerr\fi}}

\makeatother

\usepackage[normalem]{ulem}
\usepackage{tcolorbox}

\usepackage{upgreek}
\usepackage[colorlinks]{hyperref}
\usepackage[capitalize]{cleveref}

\crefname{equation}{eq.}{eqs.} 
\crefname{enumi}{}{} 

\crefname{icase}{case}{cases}
\crefname{ipart}{part}{parts}
\crefname{iprop}{property}{properties}
\crefname{iinv}{invariant}{invariants}

\crefformat{section}{\S\,#2#1#3}
\crefformat{subsection}{\S\,#2#1#3}
\crefrangeformat{section}{\S\S\,#3#1#4--#5#2#6}
\crefmultiformat{section}{\S\S\,#2#1#3}{,~#2#1#3}{,~#2#1#3}{,~#2#1#3}

\usepackage{authblk}
\usepackage[normalem]{ulem}

\usepackage{tikz}
\usetikzlibrary{calc}

\newcommand{\poly}{\mathrm{poly}}

\newcommand{\fO}{\mathcal{O}}

\renewcommand{\alpha}{\upalpha}

\newcommand{\connected}[1]{\def\temp{#1}\ifx\temp\empty\sim\else\overset{#1}{\sim}\fi}

\newtheorem{theorem}{Theorem}[section]
\newtheorem{lemma}[theorem]{Lemma}

\theoremstyle{definition}

\setparsizes{0pt}{0.25\baselineskip plus .25\baselineskip}{1em plus 1fil}

\author[]{L\'{a}szl\'{o} Kozma}

\affil[]{Dresden University of Technology}

\title{Faster \#PPM from faster \#3SAT}

\date{}

\begin{document}

\maketitle

\vspace{-0.4in}
\begin{abstract}
The counting variant (\#PPM) of the \emph{permutation pattern matching} (PPM) problem asks for the number of occurrences of a permutation pattern in a host permutation of length $n$. The problem admits faster algorithms than the trivial $\fO^*(2^n)$ time enumeration; the best known running time, for both PPM and \#PPM, is $\fO(1.415^n)$ (Gawrychowski and Rzepecki). It is also known that PPM is ETH-hard: an algorithm with running time $2^{o(n/\log{n})}$ would imply a subexponential algorithm for 3SAT (Jel\'{i}nek, Opler, and Pek\'{a}rek).

We give a reduction in the opposite direction, 
showing that an algorithm for counting satisfying assignments of a 3-CNF formula (\#3SAT) in time $2^{\fO({n}/{f(n)})}$, where $n$ is the number of variables, yields an algorithm for \#PPM with running time $2^{\fO({n \sqrt{\log{f(n)}/{f(n)}}})}$. In particular, the recently claimed refutation of ETH by OpenAI implies a \#3SAT algorithm with an improvement factor $f(n) \in \Omega(\sqrt{\log{n}}/\log\log{n})$ in the exponent; this yields, through our reduction, a subexponential algorithm for \#PPM. 
\end{abstract}

 \textbf{Note:} This is a technical report not intended for conference or journal publication. AI tools were used interactively to derive and check essential parts of the result; the writeup is by the author.

\section{Introduction}\label{sec1}

In this short note we show that a subexponential algorithm for \#3SAT (counting satisfying assignments of a $3$-CNF formula) implies a subexponential algorithm for \#PPM (counting occurrences of a permutation pattern in a host permutation). Previous work established hardness results for PPM and \#PPM under the exponential-time hypothesis (ETH), i.e., through reductions from 3SAT to PPM. Here, we give a reduction in the opposite direction.

The motivation for our study is the recent claim~\cite{openai} of a subexponential algorithm for \#3SAT (which, in particular, refutes ETH). We find the change of perspective from ``reducing for hardness'' to ``reducing for easiness'' refreshing, allowing the re-use of techniques and tricks familiar from practical SAT solving. We anticipate that results of this flavor will be established for many problems, and we see \#PPM as an interesting and simple, yet not entirely trivial, case study.

A permutation $\tau$ of size $n$ is a bijection $\tau: [n] \rightarrow [n]$, viewed alternatively as the sequence $\tau(1) \tau(2) \cdots \tau(n)$. We refer to $\tau(i)$ as the \emph{value} corresponding to \emph{index} $i$.

Given permutations $\tau$ and $\pi$ of size $n$ and $k$ respectively, with $k\leq n$, an \emph{occurrence} of $\pi$ in $\tau$ is a set of $k$ indices $S = \{u_1 < u_2 < \cdots < u_k\} \subseteq [n]$, whose corresponding values in $\tau$ are in the same relative order as the pattern $\pi$, i.e. $$\tau(u_i) < \tau(u_j) \iff \pi(i) < \pi(j) \quad \mbox{for all~~} i,j \in [k].$$

For example, $S = \{3,5,6\}$ is an occurrence of pattern $\pi =3\,1\,2$ in $\tau = 1\,3\,6\,5\,2\,4$, since the values $6,2,4$ are ordered like $3,1,2$. Another occurrence of the same pattern is $S = \{4,5,6\}$.

Let $\mathrm{occ}(\tau,\pi)$ denote the number of occurrences of $\pi$ in $\tau$. Permutation pattern matching (PPM) asks whether $\mathrm{occ}(\tau,\pi) >0$ and its counting version (\#PPM) asks to compute $\mathrm{occ}(\tau,\pi)$.

Both PPM and \#PPM are well-studied problems. A naive enumeration algorithm (check all subsets $S \subseteq [n]$ of size $k$) solves both, running in time $\fO^*(2^n)$. A series of improvements~\cite{Bruner,BKM,GR} led to the current best runtime of $\fO(1.415^n)$ for both problems, found by Gawrychowski and Rzepecki.

In the setting where the pattern size $k$ is separately considered, an algorithm with runtime $n^{k/4 + o(k)}$ was given by Berendsohn, Kozma, and Marx~\cite{BKM} for both PPM and \#PPM. For PPM, an algorithm with runtime of the form $n \cdot 2^{\fO(k^2 \log{k})}$ was found by Guillemot and Marx~\cite{GM}, with a subsequent improvement to $n \cdot 2^{\fO(k^2)}$ by Fox~\cite{Fox}. These latter results established the FPT status of PPM; note that they do not transfer to \#PPM. 

On the complexity side, Jel\'{i}nek, Opler, and Pek\'{a}rek~\cite{JOP} showed that, assuming the exponential-time hypothesis (ETH),  PPM does not admit a runtime of $2^{o(n/\log{n})}$ (the existence of such an algorithm would imply solving 3SAT in time $2^{o(n)}$). Berendsohn, Kozma, and Marx~\cite{BKM} ruled out, also under the ETH, a runtime of the form $g(k) \cdot n^{o(k/\log{k})}$ for \#PPM.

In this paper we study the complexity of \#PPM assuming the ``easiness'' of \#3SAT. 
We establish that a subexponential algorithm for \#3SAT implies a subexponential algorithm for \#PPM (and consequently for PPM), improving the previous exponential bounds. We focus on running times expressed in terms of $n$ only.\footnote{We note that the bounds parameterized by $k$ mentioned earlier remain preferable to the new bounds derived here, for smaller values of $k$, e.g., when $k \in \fO(n/\log^2{n})$.}

\begin{theorem}\label{thma}
An algorithm for \#3SAT with running time $2^{\fO(n/f(n))}$ 
implies an algorithm for \#PPM with running time $2^{\fO{(n \sqrt{\log{f(n)}/f(n)})}}$.
\end{theorem}

Here $f$ is assumed to be a non-decreasing, efficiently computable function satisfying
$2\le f(n)\le n$ for all sufficiently large $n$.
In particular, the claimed breakthrough for \#3SAT~\cite{openai} implies a function $f(n) \in \Omega( \sqrt{\log{n}}/\log\log{n} )$ in the exponent of the running time. This yields (assuming that the result holds), an algorithm for \#PPM with runtime $2^{\fO({n \log{\log{n}}/{(\log{n})}^{1/4})}}$. 

\paragraph{Idea of the reduction.} The natural idea is to encode the set $S$ of indices of an occurrence as a bitvector of length $n$, and design a circuit (eventually converted to a $3$-CNF formula) that verifies whether $S$ indeed describes a valid occurrence of $\pi$ in $\tau$. A straightforward implementation would do a pass through the text permutation $\tau$, counting for each index $i$ that is in $S$, the number of values smaller than the current value  $\tau(i)$, corresponding to already seen indices from $S$; these counts should then be compared to those of the pattern permutation $\pi$, to verify the match. This approach involves, however, operations and comparisons with possibly $\log{n}$-bit numbers, resulting in a circuit with $\Omega(n\log{n})$ variables. To take advantage of a (slightly) subexponential \#3SAT algorithm, we need a circuit of size close to linear.

To get around this issue, we split $[n]$ into blocks of size $b$. Both the space of indices and the space of values are partitioned into such blocks, and we ``guess'' for an occurrence of $\pi$ in $\tau$ how many indices $i \in S$ fall into each index-block and how many corresponding values $\tau(i)$ fall into each value-block.

Assuming these guesses are correct, verifying whether a given occurrence set $S$ (represented as a length $n$ bitvector) is correct becomes easier. Indeed, one only needs to check whether each index- and value-block has the correct number of items selected, and whether the selected indices and values indeed describe the pattern permutation $\pi$. 
For the latter task, we can aggregate the block counts offline, and the Boolean circuit only needs to work with local counts inside blocks, represented by $\fO(\log b)$ bits. A further optimization allows parts of the circuit to be reused for multiple equality tests. This gives $\fO(b)$ gates per index, for a final 3-CNF formula of size $\fO(nb)$.

For a suitably chosen value $b \approx \sqrt{f(n)}$, both the guessing (iterating through all feasible block-counts) and the \#3SAT subroutine take similar subexponential time, yielding the overall speedup. 

We describe the detailed reduction (proving Theorem~\ref{thma}) in the next section.

\section{Reduction from \#PPM to \#3SAT}

\subparagraph{Set-up.}\label{sec:ranks}

Consider a set $S\subseteq[n]$ of indices for a potential occurrence of $\pi$ in $\tau$. We define the \emph{index-rank} and \emph{value-rank} of every element $p$ in $S$, denoted $r_{\mathrm{i}}(p)$ and $r_{\mathrm{v}}(p)$, as follows:
\[
  r_{\mathrm{i}}(p)=|\{q\in S:q\le p\}|,\qquad r_{\mathrm{v}}(p)=|\{q\in S:\tau(q)\le\tau(p)\}|.
\]
Intuitively, $p$ is the $r_{\mathrm{i}}(p)$-th smallest selected index, and $\tau(p)$ is the $r_{\mathrm{v}}(p)$-th smallest selected value. 
In the earlier example, the indices $3,5,6$ in $S$ have index-ranks $1,2,3$ and value ranks $3,1,2$, which is the pattern $\pi$ itself. This holds in general, as shown in the next lemma. 
 
\begin{lemma}\label{lem:ranks}
A set $S\subseteq[n]$ with $|S|=k$ is an occurrence of $\pi$ if and only if $\pi(r_{\mathrm{i}}(p))=r_{\mathrm{v}}(p)$ for every $p\in S$.
\end{lemma}

\begin{proof}
Write $S=\{u_1<\dots<u_k\}$, so the selected indices up to $u_i$ are exactly $u_1,\dots,u_i$, with $r_{\mathrm{i}}(u_i)=i$.
Moreover, $r_{\mathrm{v}}$ numbers the selected indices $1,\dots,k$ in increasing order of their values: for all $p,q\in S$, we have $r_{\mathrm{v}}(p)<r_{\mathrm{v}}(q)$ if and only if $\tau(p)<\tau(q)$. If $S$ is an occurrence, the values $\tau(u_i)$ are ordered like the $\pi(i)$, so $r_{\mathrm{v}}(u_i)=\pi(i)$. Conversely, if $r_{\mathrm{v}}(u_i)=\pi(i)$ for all $i$, then $\tau(u_i)<\tau(u_{j})\Leftrightarrow r_{\mathrm{v}}(u_i)<r_{\mathrm{v}}(u_{j})\Leftrightarrow\pi(i)<\pi(j)$, so $S$ is an occurrence. 
\end{proof}

By Lemma~\ref{lem:ranks}, it suffices to encode each index $i \in [n]$ with a Boolean variable $x_i$, setting $x_i = 1$ exactly when $i \in S$. Then, a circuit can compute the ranks of selected indices in binary and verify $\pi(r_{\mathrm{i}}(i))=r_{\mathrm{v}}(i)$, for all $i \in S$. If $k \in \Omega(n)$, however, explicitly representing these ranks already requires $\Omega(n\log n)$ variables, leading to no improvement through a (slightly) subexponential \#3SAT solver. 

\paragraph{Blocks.} To facilitate a more compact encoding, partition $[n]$ into consecutive \emph{blocks}. Precisely, let $b\ge1$ (an integer to be chosen later) and let $B=\lceil n/b\rceil$. For $j\in[B]$, the $j$-th block is the interval $I(j)=\{(j-1)b+1,\dots,\min(jb,n)\}$. Note that the same intervals will be used both as blocks of indices and as blocks of values.  For an index $p\in[n]$, let $\alpha_{\mathrm{i}}(p)=\lceil p/b\rceil$ be the index of the block containing $p$, and $\alpha_{\mathrm{v}}(p)=\lceil\tau(p)/b\rceil$ the index of the block containing its value $\tau(p)$.

\paragraph{Profiles.}
A \emph{profile} is a pair $(s,t)$ of nonnegative integer vectors of
length $B$, intended to specify the numbers of selected entries
in each index-block and value-block. A profile is \emph{valid} if $0\le s(j),t(j)\le|I(j)|$ for every $j$, and both $s$ and $t$ sum to $k$.

A set $S\subseteq[n]$ has profile $(s,t)$ if
\[
  s(j)=|S\cap I(j)|,\qquad
  t(j)=|\{q\in S:\tau(q)\in I(j)\}|
  \quad\text{~~for every }j\in[B].
\]
Every subset of $[n]$ of size $k$ has exactly one valid profile, although a
valid profile need not be realized by any set. There are at most
$(b+1)^{2B}$ valid profiles (each of the $B$ index-blocks and $B$ value-blocks has a count between $0$ and $b$). For a valid profile $(s,t)$, let
$\operatorname{occ}_{s,t}(\tau,\pi)$ denote the number of occurrences
of $\pi$ in $\tau$ that have this profile.

\paragraph{Offsets and local counts.}
Let $(s,t)$ be a valid profile. As we pass through the text permutation $\tau$, we keep track of the total counts in earlier blocks, as follows:
\[
  \theta_{\mathrm{i}}(j)=s(1)+\dots+s(j-1),\qquad\theta_{\mathrm{v}}(j)=t(1)+\dots+t(j-1).
\]
We refer to $\theta_{\mathrm{i}}(j)$ and $\theta_{\mathrm{v}}(j)$ as the index- and value offsets of $j$. 

For a set $S$ of occurrence indices, and an index $p\in[n]$, we also need to count the entries that come before $p$, \emph{within} the block of $p$.  We refer to these as the \emph{local counts}, defined as follows:
\[
  \ell_{\mathrm{i}}(p)=|\{q\in S\cap I(\alpha_{\mathrm{i}}(p)):q<p\}|,\qquad \ell_{\mathrm{v}}(p)=|\{q\in S:\tau(q)\in I(\alpha_{\mathrm{v}}(p)),\ \tau(q)<\tau(p)\}|.
\]
Both local counts take values in $\{0,\dots,b-1\}$.

\begin{lemma}\label{lem:decomp}
If $S$ has profile $(s,t)$, then every $p\in S$ satisfies
\[
  r_{\mathrm{i}}(p)=\theta_{\mathrm{i}}(\alpha_{\mathrm{i}}(p))+\ell_{\mathrm{i}}(p)+1,\qquad r_{\mathrm{v}}(p)=\theta_{\mathrm{v}}(\alpha_{\mathrm{v}}(p))+\ell_{\mathrm{v}}(p)+1.
\]
\end{lemma}

\begin{proof}
The selected indices up to $p$ are the $\theta_{\mathrm{i}}(\alpha_{\mathrm{i}}(p))$ selected indices in earlier blocks, plus the $\ell_{\mathrm{i}}(p)$ selected indices before $p$ in its own block, plus $p$ itself. The same argument for the values gives the second identity.
\end{proof}

For example, take $\tau=3\,1\,4\,2\,5$, $\pi = 2\,3\,1$, $S=\{1,3,4\}$, and $b=2$. The blocks are $\{1,2\}$, $\{3,4\}$, $\{5\}$. The profile of $S$ is $s=(1,2,0)$ and $t=(1,2,0)$. For $p=3$, with value $\tau(3)=4$, Lemma~\ref{lem:decomp} gives $r_{\mathrm{i}}(3)=\theta_{\mathrm{i}}(2)+\ell_{\mathrm{i}}(3)+1=1+0+1=2$ and $r_{\mathrm{v}}(3)=\theta_{\mathrm{v}}(2)+\ell_{\mathrm{v}}(3)+1=1+1+1=3$, in agreement with $\pi(r_{\mathrm{i}}(3)) = \pi(2)=3$.

\paragraph{Lookup tables.}
Lemmas~\ref{lem:ranks} and~\ref{lem:decomp} allow us to formulate the occurrence condition as an equation between the two local counts of each selected index. Fix a valid profile $(s,t)$. For $p\in[n]$ and $\ell \in\{0,\dots,b-1\}$ with $\theta_{\mathrm{i}}(\alpha_{\mathrm{i}}(p))+\ell+1\le k$, let
\[
  g(p,\ell)=\pi\bigl(\theta_{\mathrm{i}}(\alpha_{\mathrm{i}}(p))+\ell+1\bigr)-\theta_{\mathrm{v}}(\alpha_{\mathrm{v}}(p))-1,
\]
and let $D(p)$ be the set of these $\ell$ with $0\le g(p,\ell)\le b-1$. In words: a selected index $p$ with local count $\ell_{\mathrm{i}}(p)=\ell$ agrees with the pattern exactly when its other local count $\ell_{\mathrm{v}}(p)$ equals $g(p,\ell)$.

\begin{lemma}\label{cor:lookup}
Let $S$ have profile $(s,t)$. Then $S$ is an occurrence of $\pi$ if and only if every $p\in S$ satisfies $\ell_{\mathrm{i}}(p)\in D(p)$ and $\ell_{\mathrm{v}}(p)=g(p,\ell_{\mathrm{i}}(p))$.
\end{lemma}

\begin{proof}
Since $|S|=s(1)+\dots+s(B)=k$, Lemma~\ref{lem:ranks} applies. By Lemma~\ref{lem:decomp}, $\theta_{\mathrm{i}}(\alpha_{\mathrm{i}}(p))+\ell_{\mathrm{i}}(p)+1=r_{\mathrm{i}}(p)\le k$, and the condition $\pi(r_{\mathrm{i}}(p))=r_{\mathrm{v}}(p)$ is equivalent to $\ell_{\mathrm{v}}(p)=g(p,\ell_{\mathrm{i}}(p))$. Since $\ell_{\mathrm{v}}(p)\in\{0,\dots,b-1\}$, this equality implies that $\ell_{\mathrm{i}}(p)\in D(p)$.
\end{proof}

\subsection{Constructing a formula for each profile}\label{sec:formula}

Fix a valid profile $(s,t)$. We construct a corresponding 3-CNF formula $F(s,t)$ with $\fO(nb)$ variables, so that the number of satisfying assignments $\#F(s,t)$ equals $\mathrm{occ}_{s,t}(\tau,\pi)$. This allows using a \#3SAT solver for counting the number of occurrences for the given profile.  
We proceed in two steps, first constructing a Boolean circuit, then converting it into a 3-CNF formula.

\paragraph{Circuit construction.}

We build a circuit $\Phi$ with one input bit $x_p$ for each index $p \in [n]$, interpreted as $x_p = 1$ if and only if $p \in S$. We define  $L=\lceil\log_2{(b+1)}\rceil$, so that the numbers $0,\dots,b$ fit into $L$ bits. 

We use AND, OR, and XOR gates of fan-in at most
two, with constants and complemented wires allowed. An $L$-bit incrementer adds a bit to an
$L$-bit number using $\fO(L)$ gates (ripple carry). (In)equality of two
$L$-bit numbers is tested using $L$ XOR gates and an OR of the outputs. Larger conjunctions and
disjunctions can be implemented as binary trees, consisting of gates of fan-in two; a conjunction
or disjunction of $q \ge 1$ inputs requires $q-1$ gates.

The circuit has the following four components:
\begin{enumerate}
  \item \textbf{Index counters.} In each block $I(j)$, we go through the indices $p$ in increasing order. The first index $p$ gets $C_{\mathrm{i}}(p)=0$, and every later index $p$ gets $C_{\mathrm{i}}(p)=C_{\mathrm{i}}(p-1)+x_{p-1}$, computed with an incrementer. For the last index $p$ of the block, we test that $C_{\mathrm{i}}(p)+x_p=s(j)$, i.e., that the count of selected indices in the block matches the count in the profile. 
  \item \textbf{Value counters.} In each block $I(j)$, we go through the values $v$ in increasing order and do the same, with counters $C_{\mathrm{v}}(v)$ and the input bit $x_{\tau^{-1}(v)}$ of the index that corresponds to value $v$. For the last value $v$ of the block, we test that $C_{\mathrm{v}}(v)+x_{\tau^{-1}(v)}=t(j)$.
\item \textbf{Lookup at $p$.}
We use a standard \emph{binary decoder circuit} with an $L$-bit input $z=(z_1,\ldots,z_L)$ and a $2^L$-bit output. The $j$-th output bit, for $0 \leq j \leq 2^L-1$, is $1$ if and only if $z$ represents the integer $j$ in binary, so exactly one output bit is $1$.

Such a decoder circuit can be constructed recursively as follows. For a single input bit $z_1$, the two outputs are simply $\neg z_1$ and $z_1$. When one more (less significant) input bit $a$ is appended, a number $j$ becomes $2j+a$. Thus each existing output $d_j$ gives rise to two new outputs,
$$
d'_{2j}=d_j\land\neg a,
\qquad
d'_{2j+1}=d_j\land a.
$$
Extending the decoder from $i-1$ to $i$ bits therefore requires $2^i$ binary AND gates, so the entire circuit uses 
$$
\sum_{i=2}^{L}2^i<2^{L+1}\in O(b)
$$
gates.
At each index $p$, construct one such decoder for $C_{\mathrm{i}}(p)$ and
another for $C_{\mathrm{v}}(\tau(p))$. Denote their relevant outputs by $c_{p,\ell}^{\mathrm{i}}$ and $c_{p,r}^{\mathrm{v}}$, where $c_{p,\ell}^{\mathrm{i}}$ indicates whether $C_{\mathrm{i}}(p)=\ell$, and $c_{p,r}^{\mathrm{v}}$ indicates whether $C_{\mathrm{v}}(\tau(p))=r$.

Compute
\[
  \operatorname{match}(p)
    =\bigvee_{\ell\in D(p)}
       \bigl(c_{p,\ell}^{\mathrm{i}}\land c_{p,g(p,\ell)}^{\mathrm{v}}\bigr),
  \qquad
  \operatorname{ok}(p)
    =\neg x_p\lor\operatorname{match}(p).
\]
The value $\mathrm{ok}(p)$ indicates whether index $p$ is valid within the given occurrence. If $p$ is not selected ($\neg x_p$), then there is nothing to check. Otherwise, $\mathrm{match}(p)$ verifies the condition of Lemma~\ref{cor:lookup}.

The empty disjunction is interpreted as false. Besides the
decoders, this uses at most $b$ binary AND gates, at most $b-1$
binary OR gates, and one final OR gate. Thus the entire lookup
at $p$ uses $\fO(b)$ gates of fan-in at most two.
  \item \textbf{Output.} A final AND of the $2B$ tests from parts 1 and 2 and the $n$ bits $\mathrm{ok}(p)$.
\end{enumerate}

\paragraph{Size.}
Parts 1 and 2 use $\fO(L)$ gates per index and per value,
part 3 uses $\fO(b)$ gates per index, and part 4 uses
$\fO(n+B)=\fO(n)$ gates. Thus
\[
  |\Phi| \in \fO(nL+nb+n) \subseteq \fO(nb).
\]
This count also includes the gates used to replace
larger conjunctions and disjunctions by binary trees.

\begin{lemma}\label{lem:circuit}
$\Phi(x)=1$ if and only if the set $S=\{p:x_p=1\}$ is an occurrence of $\pi$ with profile $(s,t)$.
\end{lemma}

\begin{proof}
No count exceeds $b\le 2^L-1$, so the incrementers never
overflow. Induction along each block gives
$C_{\mathrm{i}}(p)=\ell_{\mathrm{i}}(p)$ and $C_{\mathrm{v}}(\tau(p))=\ell_{\mathrm{v}}(p)$. Hence the tests
of parts 1 and 2 hold exactly when $S$ has profile $(s,t)$.

The decoder outputs satisfy
\[
  c_{p,\ell}^{\mathrm{i}}=1 \iff \ell_{\mathrm{i}}(p)=\ell,
  \qquad
  c_{p,r}^{\mathrm{v}}=1 \iff \ell_{\mathrm{v}}(p)=r.
\]
Consequently, $\operatorname{match}(p)=1$ if and only if
$\ell_{\mathrm{i}}(p)\in D(p)$ and
$\ell_{\mathrm{v}}(p)=g(p,\ell_{\mathrm{i}}(p))$. Thus
$\operatorname{ok}(p)=1$ if and only if either $p\notin S$,
or these two conditions hold. Together with the profile tests,
Lemma~\ref{cor:lookup} proves the claim.
\end{proof}

\subsection{From circuit to 3-CNF formula}

The standard Tseitin-transformation turns a circuit into a 3-CNF formula without changing the number of accepted inputs.

\begin{lemma}\label{lem:tseitin}
Let $\Phi$ be a circuit with $n$ inputs and size $g$. A 3-CNF formula $F$ with $n+g$ variables and at most $4g+1$ clauses can be computed in polynomial time such that $\#F=|\{x\in\{0,1\}^n:\Phi(x)=1\}|$.
\end{lemma}

\begin{proof}
The formula has a variable for each input and each gate. For each gate $z$ it has the clauses stating that the value $z$ equals the gate's function of its two signals. For OR, AND, and XOR gates this is possible with at most four clauses (of size at most $3$) per gate, also handling negated and constant inputs (details omitted). A final singleton clause asserts the output gate. The construction guarantees that for every assignment $x\in\{0,1\}^n$, the gate clauses have exactly one satisfying extension to the gate variables, namely the values computed by the corresponding gates in $\Phi$ on input $x$. This extension satisfies the final clause if and only if $\Phi(x)=1$.
\end{proof}

Applying Lemma~\ref{lem:tseitin} to $\Phi$ gives the formula $F(s,t)$.
It has
\[
  N=n+|\Phi| \in \fO(nb)
\]
variables and $\fO(nb)$ clauses, and by Lemma~2.4 its number
of satisfying assignments is
\[
  \#F(s,t)=\operatorname{occ}_{s,t}(\tau, \pi).
\]

\paragraph{Algorithm.}

The algorithm chooses $b$, builds $F(s,t)$ for every valid profile $(s,t)$, counts its satisfying assignments with the assumed \#3SAT solver, and outputs the sum of the counts. Since every occurrence has exactly one profile,
\[
  \mathrm{occ}(\tau,\pi)=\sum_{(s,t)\text{ valid}}\mathrm{occ}_{s,t}(\tau,\pi)=\sum_{(s,t)\text{ valid}}\#F(s,t),
\]
so the output is correct. Each count is at most $\binom{n}{k}\le2^n$, so each addition takes polynomial time.

\paragraph{Running time.}
There are at most $(b+1)^{2\lceil n/b\rceil}$ valid profiles, and hence calls. Each formula has $n \leq N \leq c_0\cdot n\cdot b$ variables for an absolute constant $c_0$, and polynomially many clauses. Each call takes time at most $2^{c_1 N/f(N)}\cdot\poly(n) \leq 2^{c\,n\,b /f(n)}\cdot\poly(n)$, for an absolute constant $c_1$, and $c=c_0\,c_1$. The total time is $2^{E(n,b)}\cdot\poly(n)$, where
\[
  E(n,b)=2\left\lceil\frac{n}{b}\right\rceil\log(b+1)+\frac{c\,n\,b}{f(n)}.
\]

Take $b=\lceil\sqrt{f(n)\log f(n)}\rceil$.
Then $b\le n$ and $\log(b+1) \in \fO(\log f(n))$.
Both terms of $E(n,b)$ have the same order:
\[
  2\left\lceil\frac{n}{b}\right\rceil\log(b+1)
  \in \fO\left(
      \frac{n\log f(n)}{\sqrt{f(n)\log f(n)}}
    \right)
  \subseteq \fO\left(n\sqrt{\frac{\log f(n)}{f(n)}}\right),
\]
and
\[
  \frac{c\,n\,b}{f(n)}
  \in \fO\left(
      \frac{n\sqrt{f(n)\log f(n)}}{f(n)}
    \right)
  \subseteq \fO\left(n\sqrt{\frac{\log f(n)}{f(n)}}\right).
\]
Since $2\le f(n)\le n$ for all sufficiently large $n$, we have
$n\sqrt{\log f(n)/f(n)} \in \Omega(\log n)$.
Thus the polynomial factor is absorbed into the exponential
bound, and the total running time is
\[
  2^{\fO\left(n\sqrt{\log f(n)/f(n)}\right)}.
\]

\bigskip

\small
\bibliographystyle{alpha}
\bibliography{main}

\end{document}